\documentclass{llncs}

\usepackage{amsmath}
\usepackage{graphicx}
\usepackage[title]{appendix}
\usepackage{amsmath,amssymb,amsfonts,subfigure}
\usepackage{comment}

\usepackage{lineno}
\makeatletter
\def\makeLineNumberLeft{%
  \linenumberfont\llap{\hb@xt@\linenumberwidth{\LineNumber\hss}\hskip\linenumbersep}% left line number
  \hskip\columnwidth% skip over column of text
  \rlap{\hskip\linenumbersep\hb@xt@\linenumberwidth{\hss\LineNumber}}\hss}% right line number
\leftlinenumbers% Re-issue [left] option
\makeatother

\title{Social Network Analysis of the Caste-Based Reservation System in India}

\author{Akrati Saxena\inst{1}  \and  Jaspal Singh Saini\inst{1} \and Yayati Gupta\inst{1}  \and Aishwarya Parasuram\inst{2}   \and Neeharika\inst{3} \and S.R.S. Iyengar\inst{1}}

\institute{Department of Computer Science and Engineering,
Indian Institute of Technology Ropar, India \newline
\email{\{akrati.saxena, jaspal.singh, yayati.gupta, sudarshan\}@iitrpr.ac.in}
\and
Google Alphabet \newline
\email{aparasuram@google.com}
\and
Department of CSE, India \newline
\email{neeharika@gmail.com}
\newline Corresponding Author: S. R. S. Iyengar
}

\begin{document}

\maketitle     
%\linenumbers
\begin{abstract}
It has been argued that the reservation system in India, which has existed since the time of Indian Independence (1947), has caused more havoc and degradation than progress. This being a popular public opinion, has not been based on any rigorous scientific study or research. In this paper, we revisit the cultural divide among the Indian population from a purely social network based approach. We study the distinct cluster formation that takes place in the Indian community and find that this is largely due to the effect of caste-based homophily. To study the impact of the reservation system, we define a new parameter called social distance that represents the social capital associated with each individual in the backward class. We study the changes that take place with regard to the average social distance of a cluster when a new link is established between the clusters which in its essence, is what the reservation system is accomplishing. Our extensive study calls for the change in the mindset of people in India. Although the animosity towards the reservation system could be rooted due to historical influence, hero worship and herd mentality, our results make it clear that the system has had a considerable impact on the country's overall development by bridging the gap between the conflicting social groups. The results also have been verified using the survey and are discussed in the paper.
% We study the reservation system in detail, starting from its past and observing its effect on the people.
\end{abstract}

%\keywords{Social Networks, Caste, Indian Reservation System, Social Capital}

\section{Introduction}
\subsection*{The Caste System in India}
In the context of the Indian society, caste is defined to be a Hindu hereditary class of socially equal persons, united in religion and usually following similar occupations, distinguished from other castes in the hierarchy by its relative degree of spiritual purity or pollution \cite{bayly2001caste}. It is said that the origin of the caste system is credited to the \emph{Vedas}, the mythological texts that claim the very basis of Hindu religion, according to which, the primal man destroyed himself to create a human society. The different castes were created from different parts of his body. The \emph{Brahmins} (scholar or priest class) were created from his head, the \emph{Kshatriyas} (soldier class) from his arms, the \emph{Vaishyas} (business class) from his thighs, and the \emph{Shudras} (menial labor class) from his feet. The hierarchy in the caste is determined by the descending order of importance of his body organs.

Another religious theory claims that the caste system was created from the body organs of Brahma, who is believed to be the creator of the world according to the Hindu religion. This stratification, though an obvious myth, has stayed in the Indian society since time immemorial.

\subsection*{Economic Impact of the Caste System}

%\textbf{Economic Impact of the Caste System}

Although the caste groups were supposedly divided along the lines of spiritual purity, they soon came to determine inflexible occupational roles. The downside of this system surfaced with the birth of caste-based discrimination, which is still a dominant phenomenon today. As time progressed, one's caste became intrinsically linked with one's wealth, social status, and even entry into public places. The society became largely dominated by the so called upper castes, and the rest were denied economic freedom, forced to work menial jobs and prevented from trying to improve their economic status. This led to the concentration of assets, social capital, and power in the hands of one section of the society. Therefore, while the socially forward classes progressed, the socially backward classes lagged behind in terms of literacy rates, education levels, income levels, and other measures of socio-economic well-being.

\subsection*{The Reservation System in India from a Network Theory Perspective}
To balance this routine of prejudicial social stratification, affirmative steps were undertaken to uplift the backward classes, and the reservation system was introduced wherein a certain number of seats are reserved for the members of socially and economically backward classes at the places of higher education and government jobs. However, it was soon met with a lot of backlash from the socially forward community, who felt that this system is not meritorious, and provides an undue advantage to members of the socially and economically backward classes.

In this paper, we study the existing reservation system from a pure network-theoretic perspective. The network science concepts like homophily \cite{homophily}, weak ties \cite{sowt}, social distance, opinion formation \cite{strang1998diffusion}, influence propagation \cite{welch1992sequential,banerjee1992simple,saxena2015understanding}, can help to explain the positive outcomes of the reservation system. Here, we consider two social communities, the \emph{socially forward and uplifted} (FC) and the \emph{socially backward and downtrodden} (BC). We establish why and how the reservation system maintains a very good balance between the two.

Past studies in network analytics by Jackson has shown that network formation and subsequent interaction between the nodes is highly influenced by homophily \cite{new-3}. In India, associations amongst the people are seen to be largely determined by caste-based homophily, hence, for the purposes of our study, we choose to term the network formation pattern among the Indian population as a \emph{caste-based homophilic network}.

Our motivation comes from the existence of a tangible strength associated with every weak tie, as proposed by Granovetter in his famously cited theory of the \textit{Strength of Weak Ties} \cite{sowt}. He observed the importance of weak ties to get new opportunities. This is a very prominent network phenomenon that has been ignored in past studies of the reservation system, which we have chosen as the key element of our study. We assume whenever a reservation is given, it motivates a weak link between BC and FC.  The mathematical model aids us in finding the number of links between the two communities as a measure of stability for the large scale social structure. We study the cumulative social capital of the backward classes on discrete time steps and observe how it changes when this system is in place.

%The every single reserved seat establish the links between the two communities. The mathematical model aids us in finding the optimum number of links as a measure of stability for the large scale social structure.
%We study the cumulative social capital of the backward classes in discrete time steps and observe how it changes when this system is in place.
% The mathematical model aids us in finding parameters, which quantify upliftment and betterment of the country. In addition, we consider the number of links between the two communities as a measure of stability for the large scale social structure.
%Primarily, we justify the need for this system to continue to exist in present day India by conducting a simple survey, gathering certain relevant data from a nation-wide census.

%We have studied this phenomenon using a simple survey, in which we found that majority of the friendship links existed within the same caste groups, as depicted in Table \ref{table:table1}. The survey was conducted among 150 subjects, all of whom were between 20 and 40 years of age, and comprised of an equal distribution between members of both forward and backward castes. Such a network
%structure and the resulting non-cooperative behavior is one of the major reasons for discrimination,
%which is detrimental to harmonious growth of the country.
%In this paper we consider a modified but natural model, 
The proposed model is a modified but simple and natural model, where it is common knowledge that the state of the world changes deterministically over time, as new network connections are added through time steps. As our main contribution, we introduce in this paper the prominent role played by the strength of weak ties in alleviating the divide between the caste groups in the Indian scenario. We find it sufficient to insert a minimal number of links between the two clusters, in order to foster harmonic relations between the two conflicting groups. As a long term aim of the reservation system, we see benefits reasonably distributed evenly among members from the forward as well as the backward communities. However, the current statistics show a clear tip in the balance favoring the socially forward community, with a majority of the country's shared resources such as education, wealth, and land-holdings, being in the possession of or being accessible to only one section of the society. This undesirable disparity can be seen as the result of many recent studies, including the works of Kumar and Rustagi \cite{disp-1}, Sedwal and Kamat \cite{disp-2} and Biradar \cite{isec}.
%Hence, on studying this system in detail, analyzing the changes in individual benefits that result due to addition of edges between the clusters through the reservation system and watching the strength of weak ties in play in propagating benefits even to those who do not directly benefit from the Reservation System, we note a number of striking changes that arise in individual welfare. 

The effects of the absence of the system can not be studied in the present day scenario due to the obvious reasons, however, a similar study has been performed by Borooah et al. \cite{borooah2007effectiveness}. The study took into account a social group within the country which was of the same social, educational and economic status as the Scheduled Castes (SC) and Scheduled Tribes (ST) in the pre-independence era. However, the condition of this group was observed to be at much more elevated state, proving the efficiency of the Reservation System. Many such similar comparative studies along with extant evidence only bolster the claim that this system is indeed a beacon of hope for the social disparity in India \cite{new-7,samp-1,samp-2,samp-9}.

We also conducted a survey\footnote{The opinion survey details and results are placed in appendix D. In survey, $66.47\%$ people are from non-reserved category and the remaining are from the reserved category.} among people who belong to various educational institutions following reservation system. These people are affected directly by the caste based reservation system. After garnering 1005 responses from the survey, we came up with few observations that are discussed below.
%On comparsion with results from the survey we observed two major things that show the impact of \emph{breeze effects},which will be explained in further sections, when asked students,we observed:
\begin{itemize}
\item When we inquired the FC students regarding their meritorious opinion on BC students, they told $25.11\%$ of BC students they interacted with had broken the stereotype associated with the reserved category students.
\item The second observation was that the students from BC background, who were able to gain admission into the educational institution through reservation, are able to influence $57.57\%$ of their younger siblings and friends from their community to try and gain admission into the institutions, causing genuine competition and increase of count on merit performance from BC background.
\end{itemize}
Both of these observations show the impact of \textit{Forward and Backward breeze effects} caused by the reservation system. This is discussed in detail in section 3, and the mathematical model is build based on these observations.

The rest of this paper is organized as follows. Next, we discuss the required preliminaries followed by the network-based analysis of caste reservation system. The simulation results are discussed in Section 4. The paper is concluded in Section 5. The proposed model has various future directions that are also discussed in the conclusion.
%and these are explained as \textit{Forward and Backward breeze effects} in the next section. We could quantify and observe the result generated from the survey and further compare it with the phenomenon discussed in the paper.
%\begin{itemize}
%\item 1
%\item To see how accepted backward community or reservation students are by their forward community or merit student conterparts in the said institutions.
%\item The ability or the amount of influence a reservation student had on his sibilings or friends from his community to also try to gain an oppurtunity to avail reservation.
%\end{itemize}

\section{Preliminaries and Definitions}

Let $G(V,E)$ represents the undirected social network under consideration, and  $G_1(V_1,E_1)$ and $G_2(V_2,E_2)$ represent the induced subgraphs of $G$, where $V_1$ is the set of all BC nodes and $V_2$ is the set of all FC nodes. Therefore, $V_1\cup V_2 = V$ and $V_1 \cap V_2 = \emptyset $. Let $n_1$ and $n_2$ be the shorthand notation for the number of individuals in the BC and FC respectively i.e.\ $|V_1| = n_1$ and $|V_2| = n_2$. 
 
We define the edge set $B$ as ($E - (E_1\cup E_2)$) i.e. $B$ consists of precisely those edges $\{u,v\}$, where $u \in V_1$ and $v \in V_2$, we will henceforth address these edges as \emph{bridges}.
The distance $d(u,v)$ between two nodes $u$ and $v$ represents the length of the shortest path between $u$ and $v$. For each $u \in V_1$ i.e. a person belonging to BC, we define $d^{*}_u$ as,
\begin{center}
$d^{*}_u = \mbox{ min }\{k | \exists v \in V_2 \ni d(u,v) = k\} $
\end{center}
Therefore, $d^{*}_u$ is the minimum distance from node $u$ at which it will find at least one node of $V_2$. We will refer to this parameter as the \emph{social distance} of node $u$ from the FC.

A path $\langle v_1, v_2, v_3,...,v_k \rangle$ is called an \emph{entry path} if $v_k \in V_2$ and $v_i \in V_1$ $\forall 1 \leq i \leq k-1$. Therefore, if $d^{*}_u = l$, then $l$ is the length of the shortest entry path starting from node $u$.

%%	Model 

\section{Caste Reservation System: A Network Analysis Approach}

The homophily observed in the social structure under consideration is \emph{selection} based \cite{homophily} i.e. the common characteristics that bound people together are immutable, in this case, it is the caste of an individual. What the Reservation System does in essence is, it picks a BC individual and gets it in contact with a group of closely knit FC individuals. For example, a BC student getting a seat in a university through the reservation, implicitly creates friendship ties with a group of close FC students.  The addition of such bridges has two-fold benefits, we term these the \emph{forward breeze} effect and the \emph{backward breeze} effect. The \emph{forward breeze} represents the change in the mindset of the FC, on coming in contact with the BC and the \emph{backward breeze} represents the increased motivation felt by the BC to achieve upliftment, by being influenced by FC members close to them.   

Next, our aim is to calculate the gain in the social capital of the BC as a function of the bridges added in the network. There exists no universal definition or technique for measuring social capital \cite{def-sc}, this can be attributed to the inherent subjectivity in the concept of social capital. However, in an exhaustive survey \cite{sc-survey} , the author differentiates between the social capital of different types:
\begin{enumerate}
\item social capital of an individual with respect to her position in the social network.
\item social capital of a group with respect to the underlying relationships within the group.
\item social capital of a group with respect to the network topological connections to other groups.
\end{enumerate}

In the Caste Reservation scenario, the cumulative social capital to be calculated falls under category 3. The social capital of category 3 was first studied in \cite{sc-cat3-intro}, where the author suggests that the teams with strong outside connections generally performed better compared to groups with weaker connections outside their group. Everett and Borgatti proposed a network measure termed \emph{group centrality} to quantize social capital of type 3 \cite{sc-group-cen}. We adopt a modified version of this definition, which fits well with our model. We define the cumulative social capital of BC as the linear sum of social capital of all the individuals present in BC. Further, for every individual $u$ in BC, we assume its social distance ($d^{*}_u$) to be a direct measure of its social capital. Lower the social distance of an individual $u$, higher is its social capital and vice versa.

As stated earlier, a person from BC getting reservation implies that she has an opportunity to form ties with a set of closely knit FC individuals. But for the sake of our analysis, we assume that only one tie exists per reservation i.e. all this bunch of weak ties is equivalent to one bridge while calculating the social capital of BC. This is a safe assumption, since, we are measuring social capital as a function of distance, which will rarely change for a pair of nodes in the network when we remove multiple copies of similar functioning edges. However, these multiple edges with respect to one reservation are not equivalent to one bridge in every aspect. For example, the presence of multiple weak ties amplifies the strength of the bridge across, in a sense that, even if one link breaks in the future, it does not influence the network topology or social capital significantly.

Our aim is to analyze the fall in $d^{*}_{u}$ ($u \in V_1$) as a function of the number of random bridges added in the system. Generally, the social networks depict scale free degree distribution \cite{bamodel} i.e. $P(degree (u) = k) \propto k^{-\gamma}$, where $2 < \gamma < 3$. Hence, ideally we must consider both commuities to be scale free graphs. However, for making the analysis of social distance  more tractable we will assume $G_1$ and $G_2$ to be \emph{Erdos-Renyi} random graphs \cite{erdos} with parameters $(n_1, p_1)$ and $(n_2, p_2)$.  Empirically, the social distance i.e. $d^{*}_{u}$ falls at nearly the same rate, independent of whether we consider communities to be scale free graphs or random graphs for the same number of nodes and edges, as shown in figure 1.  %$G_1$ and $G_2$

Further, for every $u \in V_1$ and $v \in V_2$, the edge $\{u,v\}$ is present with probability $b$, which we term as the \emph{bridging probability}. We define this new probability space of graphs as \emph{Coupled Erdos-Renyi Graphs} where both communities are represented using two \emph{Erdos-Renyi} random graphs and reservation links are represented as bridge edges placed between them.

%On comparsion with results from the survey we observed two major things that show the impact of \emph{breeze effects}, when asked students, from a large group of FC's, 25.11\% of BC students they interacted with where above par in expectations placed on them and where students with merit.The second observation was that, the students from BC background who were able to gain admission into educational institution through reservation where able to influence 57.57\% of their younger sibilings and frineds from their family background to try and gain admission into the institutions, causing geniune competition and increase of count on merit performance from BC background.
%\begin{figure}[h!]
%  \centering
%  	\hspace*{-1in}
%    \includegraphics[width=1.4\textwidth]{figure.png}
%  \caption{Fall in the average social distance as a function of the number of bridges in the system}
%  \label{fig1}
%\end{figure}
%%	Subsection: Mathematical Setup
\subsection{Social Distance Analysis on Coupled Erdos-Renyi Graphs}
The two classes BC and FC are represented by two \emph{Erdos-Renyi} random graphs $G_1(n_1,p_1)$ and $G_2(n_2,p_2)$ respectively, where $V(G_1)= \{1,2,3,\ldots,n_1\}$ and $V(G_2) = \{n_1+1,n_1+2,\ldots, n_1+n_2 \}$. All the results proved in this paper will be for asymptotically large graphs $G_1$ and $G_2$ i.e.\ $n_1 \rightarrow \infty$ and $n_2 \rightarrow \infty$.  Every possible edge across the two graphs ($n_1n_2$ in total) is added with the bridging probability $b$. 

Let $u$ represents an arbitrary node of BC. Our analysis is aimed to calculate $d_u^{*}$ i.e. the social distance of node $u$ from the FC. We begin by developing a few preliminary results.
%	Lemma 1
\begin{lemma}\label{lemma1}
$\dfrac{n!}{(n-l)!} \sim$\footnote{ $f(n) \sim g(n)$ if $\underset{n \rightarrow \infty}{\lim} f(n)/g(n) = 1$ } $n^l$ as $ l^2/n \sim 0$.
\end{lemma}

The proof of this lemma is provided in Appendix A. \\

Let $M_l$ represents the total number of possible entry paths of length $l$ with $u$ as one of its endpoint. The next lemma provides an approximation for the constant $M_l$ as a function of $l$. 

%	Lemma 2
\begin{lemma}\label{lemma2}
$M_l \sim n_2(n_1)^{l-1}$ as $ l^2/n_1 \sim 0$.
\end{lemma}

\begin{proof}
To construct an entry path of length $l$ with $u$ as one of its endpoint, we need a vertex from $V_2$ and a sequence of $l-1$ vertices from $V_1-u$ i.e. $1$ node is to be selected from $n_2$ nodes and $l-1$ nodes are to be selected from $n_1-1$ nodes, and these $l-1$ selected nodes can be permuted in $(l-1)!$ ways. 
\begin{align*}
\implies M_l &= {n_2 \choose 1}{n_1-1 \choose l-1} (l-1)! \\
 &= n_2 \dfrac{(n_1-1)!}{(n_1-l)!}  \\
 &\sim n_2 (n_1)^{l-1}	\tag{ from Lemma \ref{lemma1}} \\
\end{align*}
\qed
\end{proof}

Further, let $X_l$ represents a random variable which is equal to the number of entry paths of length $l$ with $u$ as one of its endpoint. Our next result calculates the average number of entry paths of length $l$ originating from the BC node $u$.

%	lemma 3
\begin{lemma}\label{lemma3}
$E[X_l] \sim (n_2b)(n_1p_1)^{l-1}$ as $ l^2/n_1 \sim 0$
\end{lemma}

\begin{proof}

\begin{align*}
X_l& = \sum_{i=1}^{M_l}Y_i	\\
\mbox{ where }	Y_i &=
\left\{
	\begin{array}{ll}
		1  & \mbox{if $j^{th}$ entry path is present} \\
		0 & \mbox{otherwise }
	\end{array}
\right.		\\
\implies E[X_l] &= \sum_{i=1}^{M_l}E[Y_i]	\tag{using linearity of expectation}\\
\end{align*}
An entry path $P=\langle v_{\alpha_1}$, $ v_{\alpha_2}$, $v_{\alpha_3}$, $\ldots$, $v_{\alpha_{l+1}} \rangle$ of length $l$ exists if $(l-1)$ edges ($\{v_{\alpha_1}$,$v_{\alpha_2}\}$,$\{v_{\alpha_2}$, $v_{\alpha_3}\}$, $\ldots$, $ \{v_{\alpha_{l-1}}$, $v_{\alpha_{l}}\}$) are present in $G_1$ and the bridge $\{v_{\alpha_l},v_{\alpha_{l+1}}\}$ is also present. Therefore, the probability that the entry path $P$ exists, is equal to $(p_1)^{l-1}b$. 
\begin{align*}
\implies E[X_l] &= M_l (p_1)^{l-1}b	 \\
 &\sim (n_2b)(n_1p_1)^{l-1}	\tag{from Lemma \ref{lemma2}}
\end{align*} 
\qed
\end{proof}

\paragraph*{•}
We are interested only in the case where $b < 1/n_2$, since for $b \geq 1/n_2$, expected number of bridges per node in BC will be greater than or equal to $1$, which is unrealistic in the Caste Reservation scenario. Henceforth, throughout the analysis, $b$ is assumed to be less than $1/n_2$.

%	theorem 1
\begin{theorem}\label{thm1} 
For a random graph $G_{n,p}$, $p_0 = \log(n)/n$ is the threshold probability for the property of connectedness. 
\end{theorem}

A detailed proof is available at \cite{bollobas}. \\

Since the two considered graphs $G_1$ and $G_2$ are connected, the above lemma provides a lower bound on $p_1$ and $p_2$ and hence on the density of the graphs $G_1$ and $G_2$. Therefore, $n_1p_1 > log(n_1)$ and $n_2p_2 > log(n_2)$. 
 
Next, we analyze the quantity $X_l$ i.e. the number of entry paths from node $u$ as a function of $l$. For small values of $l$ the number of entry paths $X_l$ will be negligible ($<<1$). Our aim is to find the smallest distance $d$ such that there exists at least one entry path of length $d$ from node $u$. We prove that distance $d$ is equal to $\log_{(n_1p_1)}(1/n_2b)+1$. Henceforth, for the sake of simplicity, we will represent the quantity $\log_{(n_1p_1)}(1/n_2b)$ by $d_0$.

%	theorem 2
\begin{theorem}\label{thm2} 
%Almost always there exists no entry path of length less than or equal to $d_0$ with $u$ as its endpoint i.e. $P(X_{i} = 0) \sim 0 $ for $1 \leq i \leq d_0$.
The probability to exist an entry path of length less than or equal to $d_0$ with $u$ as its endpoint is almost equal to zero i.e. $P(X_{i} = 0) \sim 0 $ for $1 \leq i \leq d_0$.
\end{theorem}
\begin{proof}
\begin{align*}
P(X_{i} \geq a) &\leq E[X_{i}]/a		\tag{using Markov's inequality}	\\
\implies P(X_{i} \geq 1) &\leq E[X_{i}]		\\
 &\sim (n_2b)(n_1p_1)^{i-1}	\tag{from Lemma \ref{lemma3}}	\\
  &\leq (n_2b)(n_1p_1)^{d_o-1}	\tag{Since $i < d_0$}	\\
 &= \dfrac{1}{n_1p_1}	\\
&< \dfrac{1}{\log(n_1)}		\tag{ from Theorem \ref{thm1} }	\\
\implies P(X_{i} \geq 1) &\rightarrow 0		\\
\implies P(X_{i} = 0) &\rightarrow 1
\end{align*}
\qed
\end{proof}

Further, we will prove that almost always (i.e. with probability close to one) there exists at least one entry path of length $d_0+1$ from $u$. Hence, proving our claim that $d_u^{*} = d_0+1$.  

%	Lemma 4
\begin{lemma}\label{lemma4}
For any random variable $X$, $P(X = 0) \leq \dfrac{\sigma^{2}_{X}}{\mu^2_X}$, where $\sigma_X$ and $\mu_X$ represent the variance and mean of the random variable $X$ respectively.
\end{lemma}

Its proof is discussed in Appendix B.

%	Lemma 5
\begin{lemma}\label{lemma5}
The standard deviation of the random variable $X_{d_0+1}$ approaches zero i.e. $\sigma_{X_{(d_0+1)}} \sim 0$
\end{lemma}

\begin{proof}
\begin{align*}
X_l& = \sum_{i=1}^{M_l}Y_i	\\
\implies X_l^2& = \sum_{i=1}^{M_l}\sum_{j=1}^{M_l} Y_iY_j = \sum_{k=0}^{l} Z_k	\\
%& = \sum_{k=0}^{l} Z_k	\\
\end{align*}
where $Z_k$ accounts for all $Y_iY_j$\rq s, where the $i^{th}$ and $j^{th}$ entry paths have precisely $k$ edges in common. Let $|Z_k|$ represent the number of terms in $Z_k$\rq s summation. 

\begin{align*}
|Z_0| \geq {n_1-1 \choose l-1}(l-1)! {n_2 \choose 1} {n_1 - l \choose l-1} (l-1)! {n_2 - 1 \choose 1}	
\end{align*}

If none of the vertices in the two entry paths are common, then certainly none of its edges are common either, this gives us the above inequality. 
\begin{align*}
 {n_1-1 \choose l-1}(l-1)! {n_2 \choose 1} {n_1 - l \choose l-1} (l-1)! {n_2 - 1 \choose 1} &= \dfrac{(n_1-1)!}{(n_1-2l+1)!} n_2(n_2-1)	\\ 
 &\sim \dfrac{n_2(n_2-1)}{n_1(n_1-2l+1)}n_1^{2l}	\\ 
  &\sim n_2^2n_1^{2l-2}	
%% &\sim n_2^2n_1^{2l-2}(p^{2l-2}b^2)
\end{align*}
The total number of terms in the summation of $X_l^2$ are $M_l^2$ i.e.\ approximately $n_2^2n_1^{2l-2}$. Therefore, most of the summation terms of $X_l^2$ fall into the basket of $Z_0$.
\begin{align*}
E^2[X_{d_0+1}] &\sim 1	\tag{from Lemma \ref{lemma3}}\\
\sigma_{X_{l}}^2 &= E[X_l^2] - E^2[X_d]	\tag{by definition}	\\
\implies \sigma_{X_{(d_0+1)}}^2 &= E[X_{d_0+1}^2] - E^2[X_{d_0+1}]	\\
\implies \sigma_{X_{(d_0+1)}}^2 &\sim 0
\end{align*}
\qed
\end{proof}

%	theorem 3
\begin{theorem}\label{thm3} 
Almost always there exists an entry path of length equal to $d_0+1$ with $u$ as its endpoint i.e. $X_{d_0+1} \geq 1 $.% almost always.
\end{theorem}
\begin{proof}
\begin{align*}
P(X_{(d_0+1)} = 0 ) &\leq  \dfrac{ \sigma^{2}_{X_{(d_0+1)}}}{\mu^2_{X_{(d_0+1)}} } \tag{using Lemma \ref{lemma4}}	\\
\implies P(X_{(d_0+1)} = 0 ) &\sim 0	\tag{using Lemma \ref{lemma3} and \ref{lemma5}}\\
\implies P(X_{(d_0+1)} \geq 1 ) &\sim 1 
\end{align*}
\qed
\end{proof}

Therefore, almost always $d^{*}_u = \log_{n_1p_1}(1/(n_2b)) +1$. Since, this formula is independent of $u$, almost all nodes in the BC have social distance $(d_0+1)$. Let $x$ represents the expected number of bridges added in the system.

\begin{align*}
\implies x &= n_1n_2b \\
\implies d^{*}_i &= \log_{n_1p_1}(n_1/x) +1		
\end{align*}

\begin{equation}
%\implies 
d^{*}_i = \dfrac{\log(n_1) - \log(x)}{\log(n_1p_1)} +1
\end{equation}

The above theorem proves that the social distance of any arbitrary node $i$ reduces logarithmically as a function of the number of bridges in the system. Therefore, only the first few bridges are highly effective in reducing the distance between the two communities, and the bridges that are added later on, don't bring a significant change in the social capital of an individual present in BC. 
\begin{comment}
Therefore, as we add bridges, their effectiveness in reducing $d^{*}_i$ also falls. For example, consider $d^{*}_i$ to be a function $x$ i.e.\ $d^{*}_i = d^{*}_i(x)$,
\begin{align*}
d^{*}_i(1) &= \log_{n_1p_1}(n_1) +1				\\
d^{*}_i(n_1^{0.25}) &= (3/4)\log_{n_1p_1}(n_1) +1		\\
d^{*}_i(n_1^{0.5}) &= (1/2)\log_{n_1p_1}(n_1) +1		\\
\end{align*}

\paragraph*{•}
Therefore, to reduce the distance $d^{*}_i$ by $(1/4)\log_{n_1p_1}(n_1)$ initially, we need to add $n_1^{1/4}$ bridges. Whereas, to reduce the distance by another $(1/4)\log_{n_1p_1}(n_1)$, we need to add $(n_1^{1/2} - n_1^{1/4}) >> n_1^{1/4}$  additional  bridges in the system. 

\paragraph*{•}
\end{comment}

\section{Simulation Results}

\subsection{Network Models}

We have used following synthetic network generative models to study the impact of adding bridges on the social distance between both communities.

\begin{enumerate}
\item Erdos-Renyi (ER) Model: In 1969, Erdos and Renyi proposed a model to generate random networks \cite{erdos}. In this model, there are $n$ nodes and an edge is placed between a pair of nodes with some fixed probability $p$. The mathematical analysis of the proposed model is explained for such type of networks.

\item Barabasi-Albert (BA) Model: In 1999, Barabasi and Albert observed that real world networks are not random but scale-free\footnote{The detailed model is explained in appendix C.} \cite{bamodel}. They are based on the rich-gets-richer phenomenon, where a node having higher degree has a high probability of getting new connections. The degree distribution of scale-free networks follows power law, so the probability of a node having degree k is defined as $ck^{-\gamma}$, where $c$ is a constant and $\gamma$ is the power law exponent. 
%In Barabasi Albert model, each new coming nodes make $m$ number of connections, where the probability to make connection with a node is directly proportional to its degree.
%\item Clique Containing Networks: In our survey, we analyzed that a person meets on average k number of relatives each year to whom he can infect with high probability. In this model there are cliques of k nodes and these cliques are connected with each other to make a community. Detailed model is explained in appendix.
%\item Dynamic networks: Real life social networks are dynamic. When a person is born, she starts making connections in the network and when she is dead, all her connections are lost. This is a very close model to real world network. We also use a simple model to generate dynamic networks to study the effective distance. The simulation follows following steps.
%\begin{enumerate}
%\item Both communities are started with n nodes BA model.
%\item At each time step a random node is chosen and it is deleted with probability q
%\item A new node is added with probability r.
%\item A bridge link is added with probability p.
%\end{enumerate}
%All these probabilities can be set based on the growth of rate of communities and the rate of new bridge edges in the network.
\end{enumerate}

To simulate the proposed model, both FC, as well as BC community, are generated using the same model and having the same properties like network size, density, etc. 

%While pacing the bridge we also need to understand how the links are placed in the network. For example, the families who have taken the reservation are more inclined to get the reservation again. In the present work we analyze following three methods to place links between the communities:
%\begin{enumerate}
%\item Place the links randomly
%\item Place the link randomly with probability p, and with 1-p if one of your has link.
%\item the probability to get the link is directly proportional to the neighborhood weight of the node, where each first distant neighbor will contribute 1 if he has got the link and 2-distant neighbor will contribute 0.5 if he has got the link. 
%\end{enumerate}

\subsection{Discussion}

In this section, firstly, we discuss the verification of math model using simulation. Secondly, we study how the effective social distance is reduced for different types of networks.

The verification of math model is shown in Figure \ref{fig:fig1}, where x-axis shows the number of bridges and y-axis shows the average social distance. To compute the average social distance, first the social distance is computed for each node of BC and then its average is taken. The experiment is repeated 10 times to compute the average social distance for the different number of bridges. In figure \ref{fig:fig1}, red color shows the average social distance computed using equation 1 of the proposed model and blue color shows the average social distance computed using simulation on coupled Erdos-Renyi networks.
%on the networks by placing

In the math model, the average social distance is converged to $1$ when $n1$ bridges are placed as each BC node will be directly connected to one FC node. But in coupled Erdos-Renyi network, the average social distance is not converged to 1 (close to 1) when $n_1$ bridges are placed as the bridges are placed uniformly at random and one node of BC can be connected with multiple nodes of FC.

\begin{comment}
\begin{figure}[]
  \includegraphics[width=10cm]{images/math_er_10k_0002.png}
  \caption{Verification of Math Model}
  \label{fig:fig1}
\end{figure}
\end{comment}

%In the proposed model we assumed that the each node of BC community makes only one connection with the forward community when the reservation is given to this candidate. Now we relax this assumption. In this experiment two communities are taken where each community is an erdos-renyi network of node n. When a person is given the reservation, she will make k connections with forward community people. She makes her first connection with a randomly chosen forward community people, and the remaining k-1 connections will be made with the neighbors of the chosen node. This is considered as the friendship links evolves using friends of friends phenomenon, where each node prefers to make new friends who are already friends with her friends. The simulation results are shown in figure 2, where red color line shows results for model 1 and blue color line shows results for model 2. The results verify our assumption that the number of more connections do not impact the social distance drastically as the connections are made in the local neighborhood.

As we know, real-world networks possess scale-free structure, next, we study the average social distance on scale-free BA networks. The results are shown in figure \ref{fig:fig1}, where the average degree of the network is same as the corresponding ER networks. The results show that the scale-free networks follow the same pattern for average social distance and it decreases logarithmically. The mathematical analysis of these models is left as the future work. The similar results are obtained for the networks of different sizes and densities.%, but they are not included due to the space constraint.
%2 and 4 respectively that is

\begin{comment}
\begin{figure}[]
  \includegraphics[width=10cm]{images/all_10k_0002.png}
  \caption{Verification of Math Model}
  \label{fig:fig2}
\end{figure}

\begin{figure}[]
  \includegraphics[width=10cm]{images/all_10k_0004.png}
  \caption{Verification of Math Model}
  \label{fig:fig3}
\end{figure}
\end{comment}

\begin{figure}[ht!]
     \begin{center}
        \subfigure[]{%
            \label{fig:first}
            \includegraphics[width=6cm]{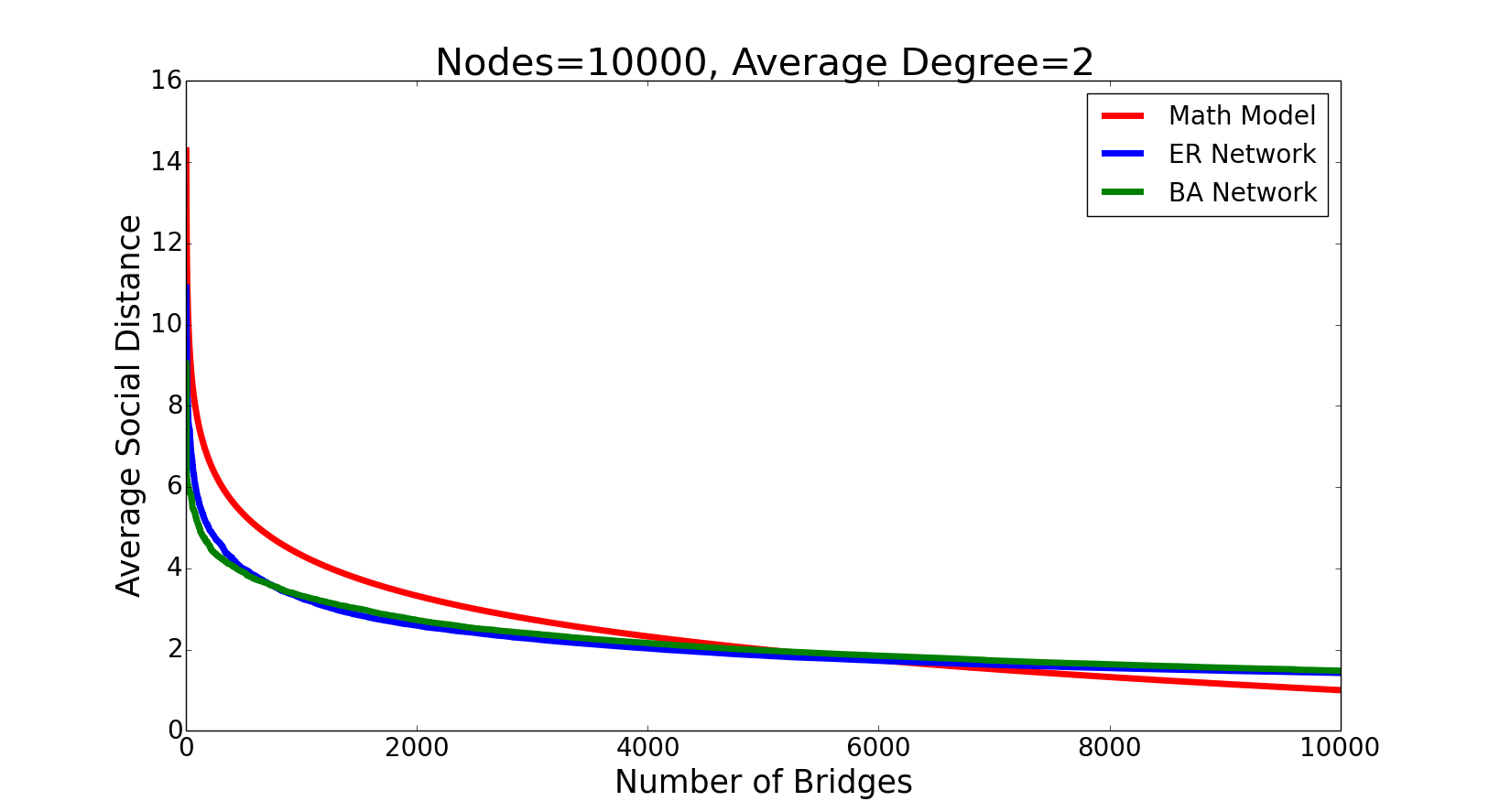}   
        }%
        \subfigure[]{%
           \label{fig:second}
           \includegraphics[width=6cm]{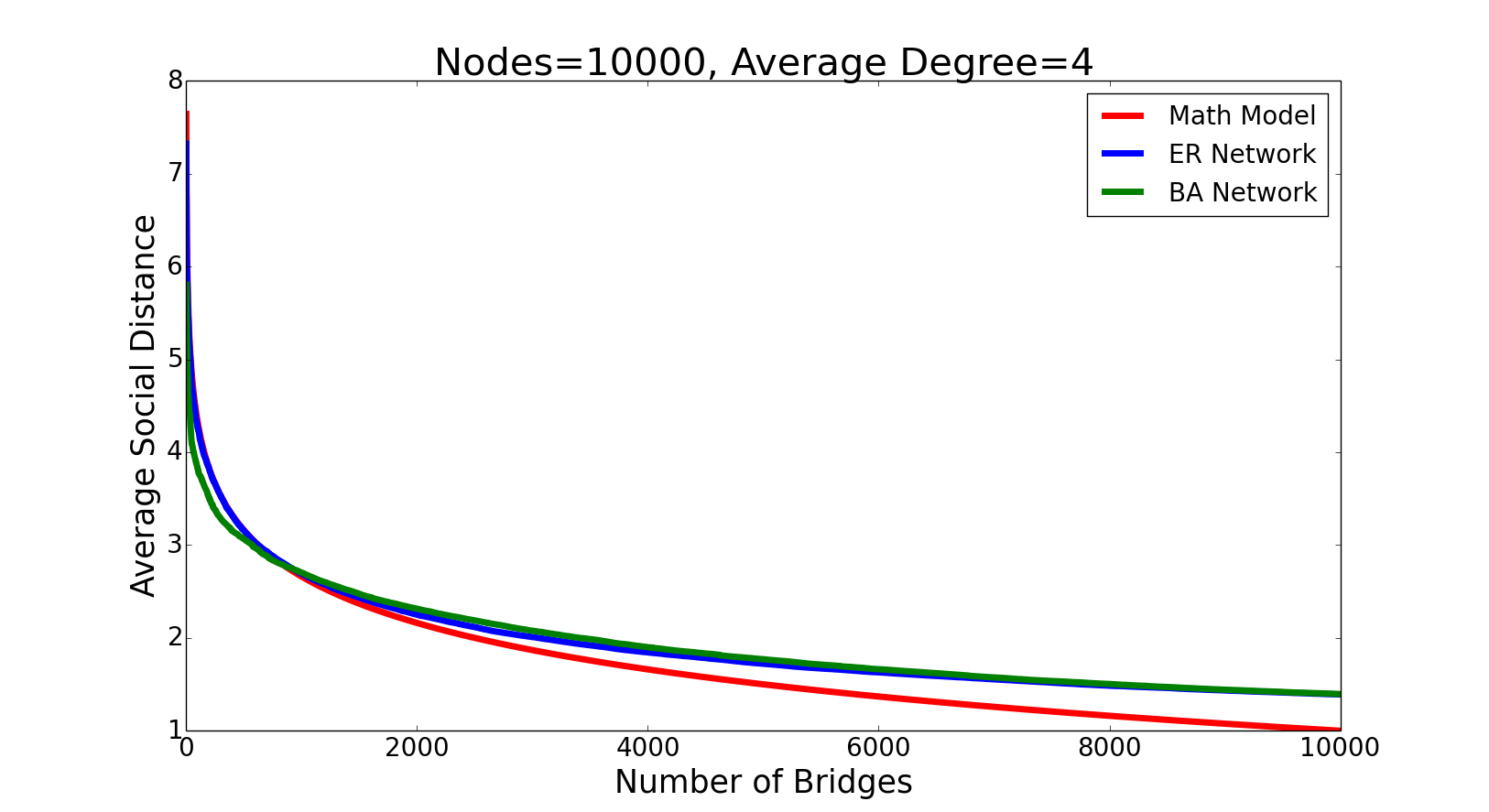}
        }
    \end{center}
    \caption{Average social distance versus number of bridges}
   \label{fig:fig1}
\end{figure}

Real-world networks posses meso-scale structures like community and core-periphery structure \cite{saxena2016evolving,gupta2016modeling}. We also simulate the proposed model on the following real-world networks.

\begin{enumerate}
%\item Facebook Network: Facebook is the most popular online social networking site today. This dataset is the induced subgraph of Facebook \cite{fb4k}, where users are represented by nodes and friendships are represented by edges. It contains 4039 nodes and 88234 edges. Results are shown in figure \ref{fig:fb}. 
\begin{comment}
\begin{figure}[]
  \includegraphics[width=10cm]{images/fb_4k.png}
  \caption{Average Social Distance on Facebook Network}
  \label{fig:fb}
\end{figure}
\end{comment}

\item Facebook Network: Facebook is the most popular online social networking website today. This dataset is the induced subgraph of Facebook \cite{facebook}, where users are represented by nodes and friendships are represented by edges. It contains 63,392 nodes and 816,831 edges. 
%\item Gowalla Social Network: This network is extracted from a location based social network called, Gowalla \cite{cho2011friendship}. This was used to shared the locations among its users. In this network, a node represents a user and an edge indicates the friendship between the user. It contains 196591 nodes and 950327 edges.
\item Twitter Network: This is an induced subgraph of Twitter \cite{twitter}. Each node is a Twitter user, and each directed edge from user A to user B means that user A follows user B. This is converted into undirected network for the study and it contains 81,306 nodes and 1,342,296 edges.
\end{enumerate}

%The results show that the social distance decreases logarithmically for real world networks.
To simulate the proposed model, two copies of real-world networks are created to represent BC and FC community. The results for Facebook and Twitter networks are shown in Figure \ref{fig:fig2}. The average social distance of real-world network is also compared with the math model by taking the same number of nodes and network densities. The results show that the average social distance decreases logarithmically in real-world networks. The bridges are placed randomly and a node can be connected with multiple bridges, so, the distance is not converged to 1 after placing $n_1$ number of bridges where $n_1$ is the size of BC community. 

\begin{figure}[ht!]
     \begin{center}
        \subfigure[]{%
            \label{fig:first}
            \includegraphics[width=6cm]{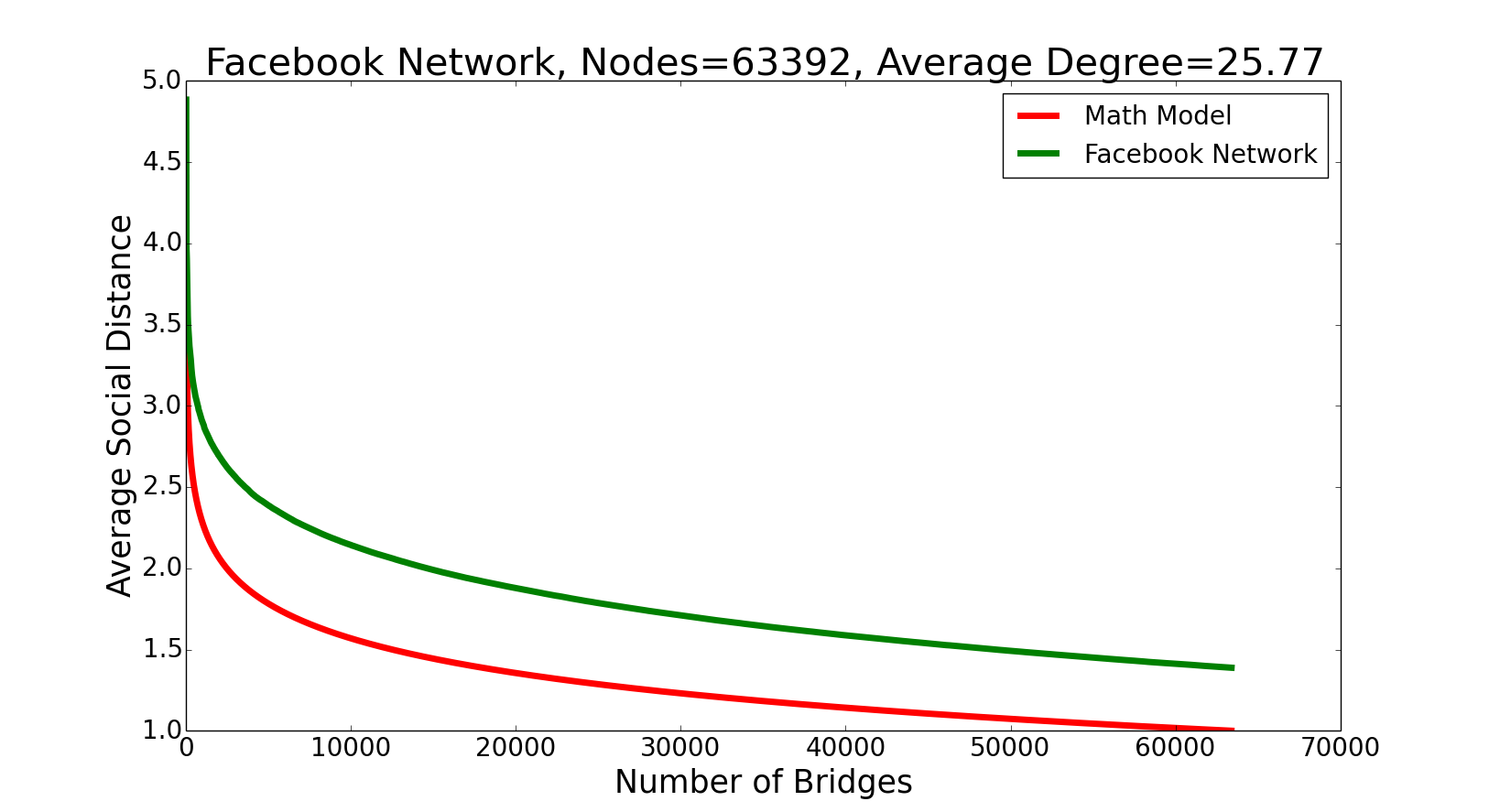}   
        }%
        \subfigure[]{%
           \label{fig:second}
           \includegraphics[width=6cm]{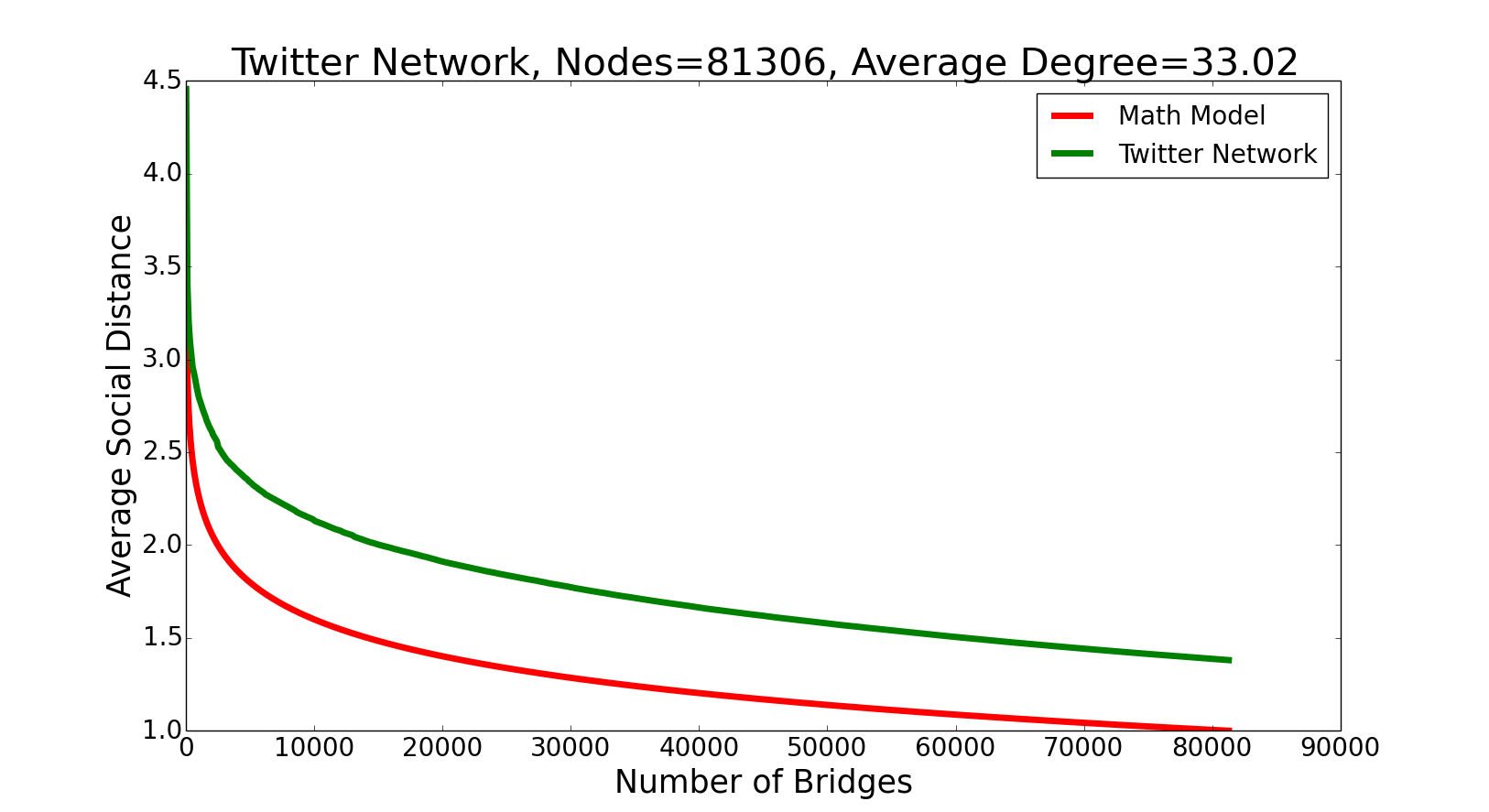}
        }
    \end{center}
    \caption{Average social distance on real world networks, a. Facebook, and b. Twitter}
   \label{fig:fig2}
\end{figure}

The simulation results support the proposed model. The results prove that the impact of placing more bridges decreases with time and a small number of bridges are sufficient to maintain the harmonic distance between the communities.

\section{Conclusion and Future Directions}
%Indian society has suffered for a long time from a discriminatory system of caste-based segregation that initially arose from concepts of spiritual purity. This led to the concentration of capital and power in the hands of a relatively small number of people, the supposed upper caste members, and polarized the society along caste-based lines of socio-economically forward and socio-economically backward groups.
%Policies of affirmative action are crucial in order to bridge this gap between the two classes.
The Indian society has suffered for a long time due to the discriminatory system of caste-based segregation that initially arose from the concepts of spiritual purity. Indian government established the reservation system to provide equal opportunities in education and employment to classes that have historically been denied access to the resources. In our paper, we looked at the reservation system from a purely network theoretic perspective, by modeling the polarized Indian society in the form of a homophilic network and considering reservation to be the phenomenon by which link formation between the two polar groups is initiated.  

We defined the social capital associated with each individual in the backward class as a function of its social distance to forward class, that quantifies an individual's access to education and employment opportunities. In the proposed model, we studied the increase in social capital of a member of the disadvantaged group as a function of the number of inter-group links (reservation opportunities). We noted that a very small number of links between the two groups are enough for the cumulative benefit to increase rapidly. To the best of our knowledge, such a model of the reservation system is the first of its kind.
%We further defined cumulative benefit, or cumulative social capital, as the sum of the individual social capital of members of the backward class.

%The proposed model has more applications where the society is divided into two parts. One example of this is the division of society based on annual income. There is one more application of the proposed method is Rural area advancement. Government can organize activities where rural people are invited to meet urban people, and it will help to reduce the social distance between these two communities. For harmonic existence of any society, it is necessary that both communities maintain the average social distance greater than the threshold. If the distance is less than this threshold than the communities can be apart and it can be harmful for the growth of the society. 

As a part of future work, we plan to investigate a larger variety of social applications wherein such a disparity exists, and apply a similar system in order to test its efficiency. This will allow us to further analyze different parameters involved in the cumulative social capital of a group. Additionally, we plan on studying the social structure within the Indian subcontinent in greater detail, and arrive at a specific percentage of reservation, called as the \emph{Ideal/Optimum Number}, which if applied, will cause great amount of uplift within shortest time, and will balance the distribution of resources between both groups.

%The proposed model can be further extended to study the impact of second level reservation. In second level reservation system, a person is applicable to get the reservation if only one of its elder generation has received the benefits of the reservation system. Once two generations of a family receive the reservation, they will not be allowed to take any further benefits. Government is still trying to understand the impact of second level reservation system and will it be really helpful for the society.

%We would like to analyze whether the reservation system should be extended for third level or it can be stopped after second level. What will be the positive and negative impact of both of these on the growth of the society.
%

We will further analyze the changes in the opinion of people towards other community that is necessary for the harmonic existence of a society. This also has been the one important motive of the reservation system. We would like to propose a mathematical model to study this phenomenon and the opinion formation in the proposed model. This will help to understand the harmonic stability in the society due to the opportunities given by the Reservation System.
% how opinion of the people changes
%is to change the opinion of people towards other community so that they can accept each other and the bars of differences are broken. We would like to propose a mathematical model to study this phenomenon, and the opinion formation in the proposed model. This will help to understand how the opinion of people changes when the people of both communities start making connections due to the opportunities given by reservation system.

%future work:
\bibliographystyle{unsrt}
\bibliography{mybib.bib}

\begin{comment}

%%%%%%%%%% BIBLIOGRAPHY %%%%%%%%%%

\end{comment}

%\pagebreak
\section*{Appendix}
\begin{subappendices}
\renewcommand{\thesection}{\Alph{section}}%
% or try \arabic{section}

\section{Proof of Lemma 1}

\begin{proof}
\begin{align*}
\dfrac{n!}{(n-l)!} &\sim \sqrt{2\pi n} \left(\dfrac{n}{e}\right)^n
 \dfrac{1}{\sqrt{2\pi (n-l)}} \left(\dfrac{e}{(n-l)}\right)^{n-l} \tag{using Stirling's approx.} 	\\
&\sim   \left(\dfrac{n}{e}\right)	^l\left(1-\dfrac{l}{n}\right)^{-(n-l+\frac{1}{2})}	\\
&\sim  \left(\dfrac{n}{e}\right)^l(e)^{l(n-l+\frac{1}{2})/n}	\\
&\sim   n^l
\end{align*}
\qed
\end{proof}

\section{Proof of Lemma 4}

\begin{proof}

\begin{align*}
P(|X - \mu_X | \geq a) &\leq \dfrac{\sigma^{2}_{X}}{a^2} \tag{Chebyshev's inequality}		\\
\implies P(|X - \mu_X | \geq \mu_X) &\leq \dfrac{\sigma^{2}_{X}}{\mu_X^2} 	 \\
 P(X=0) &\leq P(|X - \mu_X | \geq \mu_X)	\\
 &\leq \dfrac{\sigma^{2}_{X}}{\mu_X^2} \\
\end{align*}
\qed
\end{proof}

\section{Barabasi-Albert Model}
In 1999, Barabasi and Albert observed that real-world networks are not random. They observed that real-world networks follow power law degree distribution that indicates that there exists very few nodes having higher degrees and most of the nodes having lower degrees. The real world networks having power law degree distribution are also called scale-free networks. In scale-free networks, the probability $f(k)$ of a node having degree $k$ is defined as,

\begin{center}\begin{equation} \label{eq:fk}
f(k) = ck^{-\gamma}
\end{equation}\end{center}
where, $\gamma$ is the power law exponent, and for real-world scale-free networks its range is $2 < \gamma < 3$.

Based on this observation, they proposed an evolutionary preferential attachment model to generate synthetic networks that follow the properties of real-world scale-free complex networks \cite{bamodel}. This model starts with a seed graph having $n_0$ nodes that are connected with each other. At each time stamp, a new node is added to the network and it makes connections with $m$ already existing nodes. The probability $\prod(u)$ of an existing node $u$ to get a new connection is directly proportional to its degree $deg(u)$. It is defined as,

\begin{center}
$\prod(u) = \frac{deg(u)}{\sum_{\forall v}deg(v)}$
\end{center}

So, the nodes having higher degrees acquire more links over time, and the degree distribution is skewed towards lower degrees. As the network grows, only a few nodes called hubs manage to get a large number of links.

\section{Survey Details}
%add survey details here 
We conducted a survey\footnote{The survey form is available at \\ https://docs.google.com/a/iitrpr.ac.in/forms/d/1afOga9PbEQeNe3wYVFHKQ1f8JcAoyYIuOlXW3lAJD2A} amongst the very people directly involved and are affected by caste based reservation system. We gathered opinion from both ends, students who have availed the reservation and general merit students.  After garnering 1005 responses from the survey conducted across various educational institutions following reservation system, we came up with few observations corresponding to our result. In our survey, we observe how much change has already swept over both communities from either side, which is the direct result of the reservation system.
%which shows how much both communities have come closer since the introduction of the system.
%like IIT’s, NIT’s and other local institutions
All people who participated in survey belong to age 17-36.
%The survey form is shown in fig \ref{fig:form}.
Below are the questions that we asked in the survey:
\begin{itemize}
\item Please select type of institution where you have studied.
\begin{enumerate}
\item IIT/NIT/IISERS/govt Institution
\item Semi Government
\item Private
\end{enumerate}
\item Have you ever availed the reservation? Choose one of the options.
\begin{enumerate}
\item No
\item Yes (caste-reservation/management quota)
\end{enumerate}
\item Average number of relatives (family members) do you see in an year.
\item The percentage of family members that you can influence strongly with your opinion. (answer in \%)?
\item Question for NON-RESERVED Category Students: What percentage of reserved category (caste-reservation/management quota) students have impressed you with their skills?(answer in \%)?
\item Question for RESERVED Category Students: The overall percentage of your relatives (younger ones) who look up to you for inspiration to try and achieve like you. (answer in \%)?
\item Question for ALL Students (Answer it based on your view): What percentage of your family members mentioned above, can influence their friends and connections with your opinion (your academic story)? (answer in \%)?
\end{itemize}
%complete form will add later
%\begin{figure}[]
%  \includegraphics[width=10cm]{images/f1.png}
 % \caption{Verification of Math Model}
  %\label{fig:form}
%\end{figure}

\end{subappendices}

\end{document}